\documentclass[11pt, a4paper]{article}
\usepackage[T1]{fontenc}
\usepackage{amsfonts}
\usepackage{amsmath}
\usepackage{amssymb}
\usepackage{amsthm}
\usepackage{bbm}
\usepackage{bm}
\usepackage{mathrsfs}
\usepackage{verbatim}
\usepackage{setspace}
\usepackage{color}
\usepackage{pdfsync}
\usepackage{enumitem}

\theoremstyle{plain}
\newtheorem{theorem}{Theorem}[section]
\newtheorem{proposition}[theorem]{Proposition}
\newtheorem{lemma}[theorem]{Lemma}
\newtheorem{corollary}[theorem]{Corollary}

\theoremstyle{definition}

\newtheorem{remark}[theorem]{Remark}

\newtheorem{example}[theorem]{Example}

\newtheorem*{conditionA}{Condition~(A)}
\theoremstyle{remark}

\renewenvironment{thebibliography}[1]{%
\begin{oldthebibliography}{#1}%
\setlength{\baselineskip}{.9em}
\linespread{1}%
\small
\setlength{\parskip}{0ex}%
\setlength{\itemsep}{.2em}%
}%
{%
\end{oldthebibliography}%
}
\newcommand{\q}{\quad}

\newcommand{\F}{\mathbb{F}}
\newcommand{\G}{\mathbb{G}}

\newcommand{\N}{\mathbb{N}}

\newcommand{\R}{\mathbb{R}}
\renewcommand{\S}{\mathbb{S}}

\newcommand{\cA}{\mathcal{A}}
\newcommand{\cB}{\mathcal{B}}

\newcommand{\cE}{\mathcal{E}}
\newcommand{\cF}{\mathcal{F}}
\newcommand{\cG}{\mathcal{G}}
\newcommand{\cH}{\mathcal{H}}

\newcommand{\cP}{\mathcal{P}}

\newcommand{\bD}{\mathbf{D}}

\DeclareMathOperator{\esssup}{ess\, sup}

\newcommand{\as}{\mbox{-a.s.}}

\newcommand{\1}{\mathbf{1}}

\newcommand{\br}[1]{\langle #1 \rangle}

\newcommand{\tomega}{\tilde{\omega}}
\newcommand{\bomega}{\bar{\omega}}

\newcommand{\fPO}{\mathfrak{P}(\Omega)}
\newcommand{\fMa}{\mathfrak{M}_a}

\numberwithin{equation}{section}

\usepackage[pdfborder={0 0 0}]{hyperref}
\hypersetup{
  urlcolor = black,
  pdfauthor = {Ariel Neufeld, Marcel Nutz},
  pdfkeywords = {Volatility uncertainty; Superreplication; Nonlinear expectation},
  pdftitle = {Superreplication under Volatility Uncertainty for Measurable Claims},
  pdfsubject = {Superreplication under Volatility Uncertainty for Measurable Claims},
  pdfpagemode = UseNone
}

\begin{document}

\title{\vspace{-0em}
Superreplication under Volatility Uncertainty\\for Measurable Claims
\date{\today}
\author{
  Ariel Neufeld%
  \thanks{
  Dept.\ of Mathematics, ETH Zurich, \texttt{ariel.neufeld@math.ethz.ch}.
  A.~Neufeld gratefully acknowledges financial support by Swiss National Science Foundation Grant PDFMP2-137147/1.
  }
  \and
  Marcel Nutz%
  \thanks{
  Dept.\ of Mathematics, Columbia University, New York, \texttt{mnutz@math.columbia.edu}. Part of this research was carried out while M.~Nutz was visiting the Forschungsinstitut f\"ur Mathematik at ETH Zurich, and he would like to thank Martin Schweizer, Mete Soner and Josef Teichmann for their hospitality. He is also indebted to Bruno Bouchard and Jianfeng Zhang for fruitful discussions. Financial support by NSF Grant DMS-1208985 is gratefully acknowledged.
  }
 }
}
\maketitle \vspace{-1em}

\begin{abstract}
  We establish the duality formula for the superreplication price in a setting of volatility uncertainty which includes the example of ``random $G$-expectation.'' In contrast to previous results, the contingent claim is not assumed to be quasi-continuous.
\end{abstract}

\vspace{.9em}

{\small
\noindent \emph{Keywords} Volatility uncertainty; Superreplication; Nonlinear expectation

\noindent \emph{AMS 2000 Subject Classification}
93E20; %
91B30;   %
91B28 %
}

\section{Introduction}\label{se:intro}

This paper is concerned with superreplication-pricing in a setting of volatility uncertainty. We see the canonical process $B$ on the space $\Omega$ of continuous paths as the stock price process and formalize this uncertainty via a set $\cP$ of (non-equivalent) martingale laws on $\Omega$. Given a contingent claim $\xi$ measurable at time $T>0$, we are interested in determining the minimal initial capital $x\in\R$ for which there exists a trading strategy $H$ whose terminal gain $x+\int_0^T H_u\,dB_u$ exceeds $\xi$ $P$-a.s., simultaneously for all $P\in\cP$. The aim is to show that under suitable assumptions, this minimal capital is given by $x=\sup_{P\in\cP} E^P[\xi]$. We prove this duality formula for  Borel-measurable (and, more generally, upper semianalytic) claims $\xi$ and a model $\cP$ where the possible values of the volatility are determined by a set-valued process. Such a model of a ``random $G$-expectation'' was first introduced in \cite{Nutz.10Gexp}, as an extension of the ``$G$-expectation'' of~\cite{Peng.07, Peng.08}.

The duality formula under volatility uncertainty has been established for several cases and through different approaches:
\cite{DenisMartini.06} used ideas from capacity theory, \cite{Peng.10, SonerTouziZhang.2010rep, Song.10} used an approximation by Markovian control problems, and \cite{SonerTouziZhang.2010dual, NutzSoner.10} used a method discussed below. See also~\cite{PossamaiRoyerTouzi.13} for a follow-up on our results, related to optimal martingale transportation.
The main difference between our results and the previous ones is that we do not impose continuity assumptions on the claim $\xi$ (as a functional of the stock price). Thus, on the one hand, our result extends the duality formula to traded claims such as digital options or options on realized variance, which are not quasi-continuous (cf.\ \cite{DenisMartini.06}), and cases where the regularity is not known, like an American option evaluated at an optimal exercise time (cf.\ \cite{NutzZhang.12}). On the other hand, our result confirms the general robustness of the duality.

The main difficulty in our endeavor is to construct the superreplicating strategy $H$. We adopt the approach of~\cite{SonerTouziZhang.2010dual} and~\cite{NutzSoner.10}, which can be outlined as follows:

\begin{enumerate}
  \item Construct the conditional (nonlinear) expectation $\cE_t(\xi)$ related to $\cP$ and show the tower property $\cE_s(\cE_t(\xi))=\cE_s(\xi)$ for $s\leq t$.
  \item Check that the right limit $Y_t:=\cE_{t+}(\xi)$ exists and defines a supermartingale under each $P\in\cP$ (in a suitable filtration).
  \item For every $P\in\cP$, show that the martingale part in the Doob--Meyer decomposition of $Y$ is of the form $\int H^P\,dB$. Using that $H^P$ can be expressed via the quadratic (co)variation processes of $Y$ and $B$, deduce that there exists a universal process $H$ coinciding with $H^P$ under each $P$, and check that $H$ is the desired strategy.
\end{enumerate}

Step~(iii) can be accomplished by ensuring that each $P\in\cP$ has the predictable representation property. To this end---and for some details of Step~(ii) that  we shall skip for the moment---\cite{SonerTouziZhang.2010dual} introduced the set of Brownian martingale laws with positive volatility, which we shall denote by $\cP_S$: if $\cP$ is chosen as a subset of $\cP_S$, then every $P\in\cP$ has the representation property (cf.\ Lemma~\ref{le:MRPandVersions}) and Step~(iii) is feasible.

Step~(i) is the main reason why previous results required continuity assumptions on $\xi$. Recently, it was shown in~\cite{NutzVanHandel.12} that the theory of analytic sets can be used to carry out Step~(i) when $\xi$ is merely Borel-measurable (or only upper semianalytic), provided that the set of measures satisfies a condition of measurability and invariance, called Condition~(A) below (cf.\ Proposition~\ref{pr:NvH}). It was also shown in \cite{NutzVanHandel.12} that this condition is satisfied for a model of random $G$-expectation where the measures are chosen from the set of all (not necessarily Brownian) martingale laws. Thus, to follow the approach outlined above, we formulate a similar model using elements of $\cP_S$ and show that Condition~(A) is again satisfied.
This essentially boils down to proving that the set $\cP_S$ itself satisfies Condition~(A), which is our main technical result (Theorem~\ref{th:PSsatisfiesA}). Using this fact, we can go through the approach outlined above and establish our duality result (Theorem~\ref{th:duality} and Corollary~\ref{co:randomG}). As an aside of independent interest, Theorem~\ref{th:PSsatisfiesA} yields a rigorous proof for a dynamic programming principle with a fairly general reward functional (cf.\ Remark~\ref{rk:DPPforControl}).

The remainder of this paper is organized as follows. In Section~\ref{se:results}, we first describe our setup and notation in detail and we recall the relevant facts from \cite{NutzVanHandel.12}; then, we state our main results. Theorem~\ref{th:PSsatisfiesA} is proved in Section~\ref{se:proofPSsatisfiesA}, and Section~\ref{se:proofDuality} concludes with the proof of Theorem~\ref{th:duality}.

\section{Results}\label{se:results}

\subsection{Notation}

We fix the dimension $d\in\N$ and let $\Omega=\{\omega \in C([0,\infty);\R^d):\, \omega_0=0\}$ be the canonical space of continuous paths equipped with the
topology of locally uniform convergence. Moreover, let $\cF= \mathcal{B}(\Omega)$ be its Borel $\sigma$-algebra.
We denote by $B:=(B_t)_{t \geq 0}$ the canonical process $B_t(\omega)=\omega_t$, by $P_0$ the Wiener measure and by $\mathbb{F}:=(\cF_t)_{t \geq 0}$
the (raw) filtration generated by $B$. Furthermore, we denote by $\mathfrak{P}(\Omega)$ the set of all probability measures on $\Omega$, equipped with the topology of weak convergence.

We recall that a subset of a Polish space is called analytic if it is the image of a Borel subset of another Polish space under a Borel map. Moreover, an $\overline{\R}$-valued function $f$ is called upper semianalytic if $\{f>c\}$ is analytic for each $c\in\R$; in particular, any Borel-measurable function is upper semianalytic.  (See \cite[Chapter~7]{BertsekasShreve.78} for background.) Finally, the universal completion of a $\sigma$-field $\cA$ is given by
$\cA^{*}:=\cap_{P} \cA^{P}$, where $\cA^{P}$ denotes the completion with respect to $P$ and the intersection is taken over all probability measures on~$\cA$.

Throughout this paper, ``stopping time'' will refer to a finite $\F$-stopping time.  Let $\tau$ be a stopping time. Then
the concatenation of $\omega, \tomega\in \Omega$ at $\tau$ is the path
\[
 (\omega\otimes_\tau \tomega)_u := \omega_u \1_{[0,\tau(\omega))}(u) + \big(\omega_{\tau(\omega)} + \tomega_{u-\tau(\omega)}\big) \1_{[\tau(\omega), \infty)}(u),\quad u\geq 0.
\]
For any probability measure $P\in\fPO$, there is a regular conditional
probability distribution $\{P^\omega_\tau\}_{\omega\in\Omega}$
given $\cF_\tau$ satisfying
\[
  P^\omega_\tau\big\{\omega'\in \Omega: \omega' = \omega \mbox{ on } [0,\tau(\omega)]\big\} = 1\quad\mbox{for all}\quad\omega\in\Omega;
\]
cf.\ \cite[p.\,34]{StroockVaradhan.79}. We then define $P^{\tau,\omega}\in \fPO$ by
\[
  P^{\tau,\omega}(A):=P^\omega_\tau(\omega\otimes_\tau A),\quad A\in \cF, \quad\mbox{where }\omega\otimes_\tau A:=\{\omega\otimes_\tau \tomega:\, \tomega\in A\}.
\]
Given a function $\xi$ on $\Omega$ and $\omega\in\Omega$,
we also define the function $\xi^{\tau,\omega}$ on $\Omega$ by
\[
  \xi^{\tau,\omega}(\tomega) :=\xi(\omega\otimes_\tau \tomega),\q \tomega\in\Omega.
\]
Then, we have $E^{P^{\tau,\omega}}[\xi^{\tau,\omega}]=E^P[\xi|\cF_\tau](\omega)$ for $P$-a.e.\ $\omega\in\Omega$.

\subsection{Preliminaries}

We formalize volatility uncertainty via a set of local martingale laws with different volatilities. To this end, we denote by $\S$ the set of all symmetric $d\times d$-matrices and by $\S^{>0}$ its subset of strictly positive definite matrices. The set $\cP_S\subset \fPO$ consists of all local martingale laws of the form
\begin{equation}\label{eq:defPS}
  P^\alpha= P_0 \circ \bigg(\int_0^\cdot \alpha^{1/2}_s\, dB_s\bigg)^{-1},
\end{equation}
where $\alpha$ ranges over all $\F$-progressively measurable processes with values in $\mathbb{S}^{>0}$ satisfying
$\int_0^T |\alpha_s|\, ds <\infty$ $P_0$-a.s. for every $T\in\R_+$. (We denote by $|\cdot|$ the Euclidean norm in any dimension.) In other words, these are all laws of stochastic integrals of a Brownian motion, where the integrand is strictly positive and adapted to the Brownian motion. The set $\cP_S$ was introduced in~\cite{SonerTouziZhang.2010dual} and its elements have several nice properties; in particular, they have the predictable representation property which plays an important role in the proof of the duality result below (see also Section~\ref{se:proofDuality}).

We intend to model ``uncertainty'' via a subset $\cP\subset\fPO$ (below, each $P\in\cP$ will be a possible scenario for the volatility). However, for technical reasons, we make a detour and consider an entire family of subsets of $\fPO$, indexed by $(s,\omega)\in \R_+\times\Omega$, whose elements at $s=0$ coincide with $\cP$. As illustrated in Example~\ref{ex:randomG} below, this family is of purely auxiliary nature.

Let $\{\cP(s,\omega)\}_{(s,\omega)\in \R_+\times\Omega}$ be a family of subsets of $\fPO$, adapted in the sense that
\[
  \cP(s,\omega)=\cP(s,\tomega)\quad\mbox{if}\quad \omega|_{[0,s]}=\tomega|_{[0,s]},
\]
and define $\cP(\tau,\omega):=\cP(\tau(\omega),\omega)$ for any stopping time $\tau$.
Note that the set $\cP(0,\omega)$ is independent of $\omega$ as all paths start at the origin. Thus, we can define $\cP:=\cP(0,\omega)$. Before giving the example, let us state a condition on $\{\cP(s,\omega)\}$ whose purpose will be to construct the conditional sublinear expectation related to $\cP$.

\begin{conditionA}\label{condA}
  Let $s\in\R_+$, let $\tau$ be a stopping time such  that $\tau\geq s$, let $\bomega\in\Omega$ and $P\in \cP(s,\bomega)$. Set $\theta:=\tau^{s,\bomega}-s$.
  \begin{enumerate}%
    \item[(A1)] The graph $\{(P',\omega): \omega\in\Omega,\; P'\in \cP(\tau,\omega)\} \,\subseteq\, \fPO\times\Omega$ is analytic.

    \item[(A2)] We have $P^{\theta,\omega} \in\cP(\tau,\bomega\otimes_s\omega)$ for $P$-a.e.\ $\omega\in\Omega$.

    \item[(A3)] If $\nu: \Omega \to \fPO$ is an $\cF_\theta$-measurable kernel and $\nu(\omega)\in \cP(\tau,\bomega\otimes_s\omega)$ for $P$-a.e.\ $\omega\in\Omega$,
    then the measure defined by
    \begin{equation}\label{eq:defPbar}
      \bar{P}(A)=\iint (\1_A)^{\theta,\omega}(\omega') \,\nu(d\omega';\omega)\,P(d\omega),\quad A\in \cF
    \end{equation}
    is an element of $\cP(s,\bomega)$.
  \end{enumerate}
\end{conditionA}

Conditions (A1)--(A3) will ensure that the conditional expectation is measurable and satisfies the ``tower property'' (see Proposition~\ref{pr:NvH} below), which is the dynamic programming principle in this context (see~\cite{BertsekasShreve.78} for background). We remark that (A2) and (A3) imply that the family $\{\cP(s,\omega)\}$ is essentially determined by the set $\cP$. As mentioned above, in applications, $\cP$ will be the primary object and we shall simply write down a corresponding family $\{\cP(s,\omega)\}$ such that $\cP=\cP(0,\omega)$ and such that Condition~{\rm(A)} is satisfied. To illustrate this, let us state a model where the possible values of the volatility are described by a set-valued process $\bD$ and which will be the main application of our results. This model was first introduced in \cite{Nutz.10Gexp} and further studied in \cite{NutzVanHandel.12}; it generalizes the $G$-expectation of \cite{Peng.07, Peng.08} which corresponds to the case where $\bD$ is a (deterministic) compact convex set.

\begin{example}[Random $G$-Expectation]\label{ex:randomG}
  We consider a set-valued process $\bD: \Omega \times \R_+\to
  2^\S$; i.e., $\bD_t(\omega)$ is a set of matrices for each
  $(t,\omega)$.  We assume that $\bD$ is progressively graph-measurable:
  for every $t\in\R_+$,
  \[
  	\big\{(\omega,s,A)\in \Omega\times [0,t]\times\S:
  	A\in\bD_s(\omega)\big\}\in
  	\cF_t\times\cB([0,t])\times\cB(\S),
  \]
  where $\cB([0,t])$ and $\cB(\S)$ denote the Borel $\sigma$-fields of $\S$ and $[0,t]$.

  We want a set $\cP$ consisting of all $P\in\cP_S$ under which the volatility takes values in $\bD$ $P$-a.s. To this end, we introduce the auxiliary family $\{\cP(s,\omega)\}$:
  given $(s,\omega)\in\R_+\times\Omega$, we define $\cP(s,\omega)$ to be the collection of all $P\in\cP_S$ such that
  \begin{equation}\label{eq:defRandomGSets}
    \frac{d\br{B}_u}{du}(\tilde{\omega}) \in \bD_{u+s}(\omega\otimes_s\tilde{\omega}) \quad\mbox{for } P\times du\mbox{-a.e.}\; (\tilde{\omega},u) \in \Omega\times\R_+.
  \end{equation}
  In particular, $\cP:=\cP(0,\omega)$ then consists, as desired, of all $P\in\cP_S$ such that $d\br{B}_u/du \in \bD_u$ holds $P\times du$-a.e. We shall see in Corollary~\ref{co:randomG} that Condition~{\rm(A)} is satisfied in this example.
\end{example}

The following is the main result of~\cite{NutzVanHandel.12}; it is stated with the conventions
$\sup\emptyset = -\infty$ and $E^P[\xi]:=-\infty$ if $E^P[\xi^+]=E^P[\xi^-]=+\infty$, and $\esssup^P$ denotes the essential supremum under $P$.

\begin{proposition}\label{pr:NvH}
   Suppose that $\{\cP(s,\omega)\}$ satisfies Condition~{\rm(A)} and that $\cP\neq\emptyset$. Let $\sigma\leq\tau$ be stopping times and let $\xi:\Omega\to\overline{\R}$ be an upper semianalytic function. Then the function
   \[
     \cE_\tau(\xi)(\omega):=\sup_{P\in\cP(\tau,\omega)} E^P[\xi^{\tau,\omega}],\quad\omega\in\Omega
   \]
   is $\cF_\tau^*$-measurable and upper semianalytic. Moreover,
   \begin{equation}\label{eq:DPP}
     \cE_\sigma(\xi)(\omega) = \cE_\sigma(\cE_\tau(\xi))(\omega)\quad\mbox{for all}\quad \omega\in\Omega.
   \end{equation}
   Furthermore, with $\cP(\sigma;P)=\{P'\in \cP:\, P'=P \mbox{ on } \cF_\sigma\}$, we have
   \begin{equation}\label{eq:esssupDPP}
     \cE_\sigma(\xi) = \mathop{\esssup^P}_{P'\in \cP(\sigma;P)} E^{P'}[\cE_\tau(\xi)|\cF_\sigma]\quad P\as\quad\mbox{for all}\quad P\in\cP.
   \end{equation}
\end{proposition}

\subsection{Main Results}

Some more notation is needed to state our duality result. In what follows, the set $\cP$ determined by the family $\{\cP(s,\omega)\}$ will be a subset of $\cP_S$. We shall use a finite time horizon $T\in\R_+$ and the filtration $\G= (\cG_t)_{0\leq t \leq T}$, where
\begin{equation*}
  \cG_t:= \cF^{*}_{t}\vee \mathcal{N}^{\mathcal{P}}; %
\end{equation*}
here $\cF^{*}_{t}$ is the universal completion of $\cF_t$ and $\mathcal{N}^{\mathcal{P}}$ is the collection of sets which are $(\cF_T,P)$-null for all $P\in\cP$.

Let $H$ be a $\G$-predictable process taking values in $\R^d$ and such that $\int_0^T H_u^\top\,d\br{B}_u\,H_u <\infty$ $P$-a.s.\ for all $P\in\cP$. Then $H$ is called an admissible trading strategy if
\footnote{
 Here $\int H\,dB$ is, with some abuse of notation, the usual It\^o integral under the fixed measure $P$. We remark that we could also define the integral simultaneously under all $P\in\cP$ by the construction of~\cite{Nutz.11int}. This would yield a cosmetically nicer result, but we shall avoid the additional set-theoretic subtleties as this is not central to our approach.
}
$\int H\,dB$ is a $P$-supermartingale for all $P\in\cP$; as usual, this is to rule out doubling strategies. We denote by $\cH$ the set of all admissible trading strategies.

\begin{theorem}\label{th:duality}
  Suppose $\{\cP(s,\omega)\}$ satisfies Condition~{\rm(A)} and $\emptyset\neq\cP\subset \cP_S$. Moreover, let
  $\xi:\Omega\to\overline{\R}$ be an upper semianalytic, $\cG_T$-measurable function such that  $\sup_{P\in\cP} E^P[|\xi|]<\infty$.
  Then
  \begin{align*}
    &\sup_{P\in\cP} E^P[\xi] \\
    &=\min\left\{x\in\R:\,\exists\, H\in \cH\mbox{ with } x+ \int_0^T H_u\,dB_u\geq \xi\;P\as\mbox{ for all }P\in\cP\right\}.
  \end{align*}
\end{theorem}

The assumption that $\cP\subset\cP_S$ will be essential for our proof, which is stated in Section~\ref{se:proofDuality}.
In order to have nontrivial examples where the previous theorem applies, it is essential to know that the set $\cP_S$ (seen as a constant family $\cP(s,\omega)\equiv \cP_S$) satisfies itself Condition~{\rm(A)}. This fact is our main technical result.

\begin{theorem}\label{th:PSsatisfiesA}
  The set $\mathcal{P}_S$ satisfies Condition~{\rm(A)}.
\end{theorem}

The proof is stated Section~\ref{se:proofPSsatisfiesA}. It is easy to see that if two families satisfy Condition~{\rm(A)}, then so does their intersection. In particular, we have:

\begin{corollary}\label{co:intersectionSatisfiesA}
  If $\{\cP(s,\omega)\}$ satisfies Condition~{\rm(A)}, so does $\{\cP(s,\omega)\cap\cP_S\}$.
\end{corollary}

The following is the main application of our results.

\begin{corollary}\label{co:randomG}
  The family $\{\cP(s,\omega)\}$ related to the random $G$-expectation (as defined in Example~\ref{ex:randomG}) satisfies Condition~{\rm(A)}. In particular, the duality result of Theorem~\ref{th:duality} applies in this case.
\end{corollary}

\begin{proof}
  Let $\fMa\subset\fPO$ be the set of all local martingale laws on $\Omega$ under which the quadratic variation of $B$ is absolutely continuous with respect to the Lebesgue measure; then $\cP_S\subset\fMa$. Moreover, let $\tilde{\cP}(s,\omega)$ be the set of all $P\in\fMa$ such
  that~\eqref{eq:defRandomGSets} holds. Then, clearly, $\cP(s,\omega)=\tilde{\cP}(s,\omega)\cap \cP_S$, and since we know from \cite[Theorem~4.3]{NutzVanHandel.12} that $\{\tilde{\cP}(s,\omega)\}$ satisfies Condition~{\rm(A)}, Corollary~\ref{co:intersectionSatisfiesA} shows that $\{\cP(s,\omega)\}$ again satisfies Condition~{\rm(A)}.
\end{proof}

\begin{remark}\label{rk:DPPforControl}
  In view of~\eqref{eq:DPP}, Theorem~\ref{th:PSsatisfiesA} yields the dynamic programming principle for the optimal control problem
  $\sup_\alpha E^{P_0}[\xi(X^\alpha)]$ with a very general reward functional $\xi$, where $X^\alpha=\int_0^\cdot \alpha^{1/2}_s\,dB_s$. We remark that the arguments in the proof of Theorem~\ref{th:PSsatisfiesA} could be extended to other control problems; for instance, the situation where the state process $X^\alpha$ is defined as the solution of a stochastic functional/differential equation as in~\cite{Nutz.11}.
\end{remark}

\section{Proof of Theorem~\ref{th:PSsatisfiesA}}\label{se:proofPSsatisfiesA}

In this section, we prove that $\cP_S$ (i.e., the constant family $\cP(s,\omega)\equiv\cP_S$) satisfies Condition~{\rm(A)}. Up to some minor differences in notation, property (A2) was already shown in~\cite[Lemma~4.1]{SonerTouziZhang.2010dual}, so we focus on~(A1) and~(A3).

Let us fix some more notation.  We denote by $E[\,\cdot\,]$ the expectation under the Wiener measure $P_0$; more generally, any notion related to $\Omega$ that implicitly refers to a measure will be understood to refer to $P_0$. Unless otherwise stated, any topological space is endowed with its Borel $\sigma$-field. As usual,
$L^0(\Omega;\R)$ denotes the set of equivalence classes of random variables on $\Omega$, endowed with the topology of convergence in measure. Moreover, we denote by $\bar{\Omega}=\Omega\times\R_+$ the product space; here the measure is $P_0\times dt$ by default, where $dt$ is the Lebesgue measure.
The basic space in this section is $L^0(\bar{\Omega}; \mathbb{S})$, the set of equivalence classes of $\S$-valued processes that are product-measurable.
We endow $L^0(\bar{\Omega};\S)$ (and its subspaces) with the topology of local convergence in measure; that is, the metric
\begin{equation}\label{eq:defConvLocallyInMeasure}
  d(\cdot, \cdot)= \sum_{k\in\N} 2^{-k}\, \frac{d_k(\cdot, \cdot)}{1+d_k(\cdot, \cdot)}\,, \quad \mbox{where} \quad d_k(X, Y) = E\bigg[\int_0^k 1\wedge |X_s -Y_s |\, ds \bigg].
\end{equation}
As a result, $X^n\to X$ in $L^0(\bar{\Omega};\S)$ if and only if $\lim_n E [\int_0^T 1\wedge |X^n_s -X_s |\, ds ] = 0$ for all $T\in \R_+$.

\subsection{Proof of (A1)}

The aim of this subsection is to show that $\mathcal{P}_S\subset \mathfrak{P}(\Omega)$ is analytic. To this end, we shall show that $\cP_S\subset \mathfrak{P}(\Omega)$ is the image of a Borel space (i.e., a Borel subset of a Polish space) under a Borel map; this implies the claim by \cite[Proposition~7.40, p.\,165]{BertsekasShreve.78}.
Indeed, let $L_{prog}^0(\bar{\Omega};\S)\subset L^0(\bar{\Omega};\S)$ be the subset of $\F$-progressively measurable processes
and
\begin{align*}
L^1_{loc}(\bar{\Omega};\mathbb{S}^{>0})=\Big\{ \alpha \in L_{prog}^0(\bar{\Omega}; \mathbb{S}^{>0}):\,
\int_0^T |\alpha_s|\, ds < \infty \ P_0\mbox{-a.s.\ for all} \ T\in\R_+ \Big\}.
\end{align*}
Moreover, we denote by
\begin{equation}\label{eq:defMapPalpha}
  \Phi: L^1_{loc}(\bar{\Omega};\mathbb{S}^{>0})\to\mathfrak{P}(\Omega), \quad \alpha\mapsto P^\alpha=P_0 \circ \left(\int_0^\cdot \alpha^{1/2}_s\, dB_s\right)^{-1}
\end{equation}
the map which associates to $\alpha$ the corresponding law. Then $\cP_S$ is the image
\[
 \cP_S= \Phi(L^1_{loc}(\bar{\Omega};\mathbb{S}^{>0}));
\]
therefore, the claimed property (A1) follows from the two subsequent lemmas.

\begin{lemma}\label{le:L1locBorelSpace}
  The space $L_{prog}^0(\bar{\Omega};\S)$ is Polish and $L^1_{loc}(\bar{\Omega};\mathbb{S}^{>0})\subset L_{prog}^0(\bar{\Omega};\S)$ is Borel.
\end{lemma}

\begin{proof}
  We start by noting that $L^0(\bar{\Omega};\S)$ is Polish. Indeed, as $\R_+$ and $\Omega$ are separable metric spaces, we have that $L^2(\Omega\times [0,T];\S)$ is separable for all $T\in\N$ (e.g., \cite[Section~6.15, p.\,92]{Doob.94}). A density argument and the definition~\eqref{eq:defConvLocallyInMeasure} then show that $L^0(\bar{\Omega};\S)$ is again separable. On the other hand, the completeness of $\S$ is inherited by $L^0(\bar{\Omega};\S)$; see, e.g., \cite[Corollary~3]{Chung.81}.
  Since $L_{prog}^0(\bar{\Omega};\S)\subset L^0(\bar{\Omega};\S)$ is closed, it is again a Polish space.

  Next, we show that
  $L^1_{loc}(\bar{\Omega};\mathbb{S})$ is a Borel subset of $L_{prog}^0(\bar{\Omega};\S)$.
  We first observe that
  \[
    L^1_{loc}(\bar{\Omega};\mathbb{S}) = \bigcap_{T\in\N} \left\{\alpha\in L_{prog}^0(\bar{\Omega};\S):\, P_0\left[ \arctan \left(\int_0^T |\alpha_s|\,ds\right) \geq \frac{\pi}{2}\right]=0\right\}.
  \]
  Therefore, it suffices to show that for fixed $T\in\N$,
  \[
    \alpha\mapsto P_0\left[ \arctan \left(\int_0^T |\alpha_s|\,ds\right) \geq \frac{\pi}{2}\right]
  \]
  is Borel. Indeed, this is the composition of the function
  \[
    L^0(\Omega;\R)\to \R , \quad f\mapsto P_0\left[ f\geq \pi/2 \right],
  \]
  which is upper semicontinuous by the Portmanteau theorem and thus Borel, with the map
  \[
    L_{prog}^0(\bar{\Omega};\S)\to L^0(\Omega;\R), \quad  \alpha\mapsto \arctan \left(\int_0^T |\alpha_s|\,ds\right).
  \]
  The latter is Borel because it is, by monotone convergence, the pointwise limit  of the maps
  \[
    \alpha\mapsto \arctan \left(\int_0^T n\wedge |\alpha_s|\,ds\right),
  \]
  which are continuous for fixed $n\in\N$ due to the elementary estimate
  \begin{multline*}
  E\bigg[1\wedge\bigg|\arctan\bigg(\int_0^T n\wedge|\alpha_s| \, ds\bigg) - \arctan\bigg(\int_0^T n\wedge|\tilde{\alpha}_s| \, ds\bigg) \bigg| \bigg]\\
  \leq E \bigg[ \int_0^T  n\wedge |\alpha_s - \tilde{\alpha}_s |\ ds \bigg].
  \end{multline*}
  This completes the proof that $L^1_{loc}(\bar{\Omega};\mathbb{S})$ is a Borel subset of $L_{prog}^0(\bar{\Omega};\S)$. To deduce the same property for $L^1_{loc}(\bar{\Omega};\mathbb{S}^{>0})$, note that
  \[
    L^1_{loc}(\bar{\Omega};\mathbb{S}^{>0}) = \bigcap_{T\in\N} \big\{\alpha\in L^1_{loc}(\bar{\Omega};\mathbb{S}):\, \mu_T\big[ \alpha \in \S\setminus\S^{>0}\big]=0    \big\},
  \]
  where $\mu_T$ is the product measure $\mu_T(A)= E[\int_0^T\1_A\,ds]$. As $\S^{>0}\subset\S$ is open, $\alpha\mapsto \mu_T[ \alpha \in \S\setminus\S^{>0}]$ is upper semicontinuous and we conclude that
  $L^1_{loc}(\bar{\Omega};\mathbb{S}^{>0})$ is again Borel.
\end{proof}

\begin{lemma}\label{le:PalphaBorel}
  The map $\Phi: L^1_{loc}(\bar{\Omega};\mathbb{S}^{>0}) \to \mathfrak{P}(\Omega)$ defined in~\eqref{eq:defMapPalpha} is Borel.
\end{lemma}

\begin{proof}
  Consider first, for fixed $n\in\N$, the mapping $\Phi_n$ defined by
  \[
  \Phi_n(\alpha)= P_0 \circ \left(\int_0^\cdot \pi_n(\alpha^{1/2}_s)\, dB_s\right)^{-1},
  \]
  where $\pi_n$ is the projection onto the ball of radius $n$ around the origin in $\S$. It follows from a direct extension of the dominated convergence theorem for stochastic integrals \cite[Theorem~IV.32, p.\,176]{Protter.05} that
  \[
    \alpha\mapsto \int_0^\cdot \pi_n(\alpha^{1/2}_s)\, dB_s
  \]
  is continuous for the topology of uniform convergence on compacts in probability (``ucp''), and hence that $\Phi_n$ is continuous  for the topology of weak convergence. In particular, $\Phi_n$ is Borel. On the other hand, a second application of dominated convergence shows that
  \[
    \int_0^\cdot \pi_n(\alpha^{1/2}_s)\, dB_s \to \int_0^\cdot \alpha^{1/2}_s\, dB_s \quad \mbox{ucp}\quad \mbox{as $n\to\infty$}
  \]
  for all $\alpha\in L^1_{loc}(\bar{\Omega};\mathbb{S}^{>0})$ and hence that $\Phi(\alpha)=\lim_n \Phi_n(\alpha)$ in $\mathfrak{P}(\Omega)$ for all $\alpha$. Therefore, $\Phi$ is again Borel.
\end{proof}

\subsection{Proof of (A3)}

Given a stopping time $\tau$, a measure $P\in\cP_S$ and
an $\cF_\tau$-measurable kernel $\nu: \Omega \rightarrow \mathfrak{P}(\Omega)$ with $\nu(\omega) \in \mathcal{P}_S$ for $P$-a.e.\
$\omega \in \Omega$, our aim is to show that
\begin{equation*}
 \bar{P}(A) := \iint (\textbf{1}_A)^{\tau, \omega}(\omega')\, \nu(\omega, d\omega')\, P(d\omega), \quad \ A \in \cF
\end{equation*}
defines an element of $\mathcal{P}_S$. That is, we need to show that $\bar{P}=P^{\bar{\alpha}}$ for some $\bar{\alpha} \in L^1_{loc}(\bar{\Omega};\mathbb{S}^{>0})$. In the case where $\nu$ has only countably many values, this can be accomplished by explicitly writing down an appropriate process $\bar{\alpha}$, as was shown already in~\cite{SonerTouziZhang.2010dual}. The present setup requires general kernels and a measurable selection proof. Roughly speaking, up to time $\tau$, $\bar{\alpha}$ is given by the integrand $\alpha$ determining $P$, whereas after $\tau$, it is given by the integrand of $\nu(\omega)$, for a suitable $\omega$. In Step~1 below, we state precisely the requirement for $\bar{\alpha}$; in Step~2, we construct a measurable selector for the integrand of $\nu(\omega)$; finally, in Step~3, we show how to construct the required process $\bar{\alpha}$ from this selector.

\medskip

\emph{Step 1.} Let $\alpha\in L^1_{loc}(\bar{\Omega};\mathbb{S}^{>0})$ be such that $P=P^\alpha$, let $X^\alpha:=\int_0^\cdot \alpha^{1/2}_s\, dB_s$, and let $\tilde{\tau}:=\tau\circ X^\alpha$. Suppose we have found $\hat{\alpha}\in L^0_{prog}(\bar{\Omega};\mathbb{S})$ such that
\[%
  \hat{\alpha}^\omega:=\hat{\alpha}_{\cdot + \tilde{\tau}(\omega)}(\omega\otimes_{\tilde{\tau}}\cdot)\in
 L^1_{loc}(\bar{\Omega};\mathbb{S}^{>0}) \mbox{ and } P^{\hat{\alpha}^\omega}=\nu(X^\alpha(\omega)) \mbox{ for $P_0$-a.e.\ $\omega\in \Omega$.}
\]%
Then $\bar{P}=P^{\bar{\alpha}}$ for $\bar{\alpha}$ defined by
\begin{equation*}%
  \bar{\alpha}_s(\omega)
  = \alpha_s(\omega) \1_{[0,\tilde{\tau}(\omega)]}(s) + \hat{\alpha}_{s}(\omega)\1_{(\tilde{\tau}(\omega),\infty)}(s).
\end{equation*}

Indeed, we clearly have $\bar{\alpha}\in L^1_{loc}(\bar{\Omega};\mathbb{S}^{>0})$. Moreover, we note that $\tilde{\tau}$ is again a stopping time by Galmarino's test \cite[Theorem~IV.100, p.\,149]{DellacherieMeyer.78}.
To see that $\bar{P}=P^{\bar{\alpha}}$, it suffices to show that
\begin{equation*}
 E^{\bar{P}}\big[\psi\big(B_{t_1},\dots, B_{t_n}\big)\big]= E^{P_0}\big[\psi\big(X^{\bar{\alpha}}_{t_1},\dots, X^{\bar{\alpha}}_{t_n}\big)\big]
\end{equation*}
for all $n\in\N$, $0<t_1<t_2<\dots<t_n<\infty$,  and any bounded continuous function $\psi:\R^n\to \R$. Recall that $B$ has stationary and independent increments under the Wiener measure $P_0$.
For $P_0$-a.e.\ $\omega\in\Omega$ such that $\tilde{t}:=\tilde{\tau}(\omega) \in [t_j, t_{j+1})$, we have
\begin{multline*}
  E^{P_0}\big[\psi\big(X^{\bar{\alpha}}_{t_1},\dots, X^{\bar{\alpha}}_{t_n}\big)\big|\cF_{\tilde{\tau}}\big](\omega)\\
  \shoveleft{=E^{P_0}\big[\psi\big(X^{\bar{\alpha}}_{t_1}(\omega\otimes_{\tilde{t}} B),\dots, X^{\bar{\alpha}}_{t_n}(\omega\otimes_{\tilde{t}} B)\big)\big]}\\
  \shoveleft{ = E^{P_0}\bigg[ \psi\bigg(X^{\alpha}_{t_1}(\omega),\dots, X^{\alpha}_{t_j}(\omega), X^{\alpha}_{\tilde{t}}(\omega)
   + \int_{\tilde{t}}^{t_{j+1}} \hat{\alpha}_{s}^{1/2}(\omega\otimes_{\tilde{t}}B)\, dB_{s-\tilde{t}},\dots,}\\
    X^{\alpha}_{\tilde{t}}(\omega) + \int_{\tilde{t}}^{t_{n}}\hat{\alpha}_{s}^{1/2}(\omega\otimes_{\tilde{t}} B)\, dB_{s-\tilde{t}}\bigg)\bigg]
\end{multline*}
and thus, by the definition of $\alpha^{\omega}$,
\begin{multline*}
  E^{P_0}\big[\psi\big(X^{\bar{\alpha}}_{t_1},\dots, X^{\bar{\alpha}}_{t_n}\big)\big|\cF_{\tilde{\tau}}\big](\omega)\\
  \shoveleft{ = E^{P^{\alpha^{\omega}}} \big[\psi\big(X^{\alpha}_{t_1}(\omega),\dots, X^{\alpha}_{t_j}(\omega),  X^{\alpha}_{\tilde{t}}(\omega) +B_{t_{j+1}-\tilde{t}},\dots,
   X^{\alpha}_{\tilde{t}}(\omega) +B_{t_{n}-\tilde{t}}\big)\big] } \\
  \shoveleft{ = \int \psi\big(X^{\alpha}_{t_1}(\omega),\dots, X^{\alpha}_{t_j}(\omega),
   X^{\alpha}_{\tilde{t}}(\omega) +B_{t_{j+1}-\tilde{t}}(\omega'),\dots, }\\
   X^{\alpha}_{\tilde{t}}(\omega) +B_{t_{n}-\tilde{t}}(\omega')\big)\,\nu\big(X^{\alpha}(\omega), d\omega'\big).
\end{multline*}
Integrating both sides with respect to $P_0(d\omega)$ and noting that $\tilde{t}\in [t_j, t_{j+1})$ implies $t:={\tau}(\omega) \in [t_j, t_{j+1})$ $P$-a.s., we conclude that
\begin{multline*}
E^{P_0}\big[\psi\big(X^{\bar{\alpha}}_{t_1},\dots, X^{\bar{\alpha}}_{t_n}\big)\big] \\
\shoveleft{ = \iint \psi\big(X^{\alpha}_{t_1}(\omega),\dots, X^{\alpha}_{t_j}(\omega),
   X^{\alpha}_{\tilde{t}}(\omega) +B_{t_{j+1}-\tilde{t}}(\omega'),\dots, }\\
\shoveright{X^{\alpha}_{\tilde{t}}(\omega) +B_{t_{n}-\tilde{t}}(\omega')\big)\, \nu(X^\alpha(\omega), d\omega')\, P_0(d\omega) }\\
\shoveleft{  = \iint \psi\big(B_{t_1}(\omega),\dots, B_{t_j}(\omega), B_{t}(\omega) +B_{t_{j+1}-t}(\omega'),\dots, }\\
\shoveright{ B_{t}(\omega) +B_{t_{n}-t}(\omega')\big)\, \nu(\omega, d\omega')\, P(d\omega) }\\
\shoveleft{= \iint \psi^{\tau, \omega} \big(B_{t_1},\dots, B_{t_n}\big)(\omega') \, \nu(\omega, d\omega')\, P(d\omega)}\\
\shoveleft{ = E^{\bar{P}}\big[\psi\big(B_{t_1},\dots, B_{t_n}\big)\big]. }\\[-1em]
\end{multline*}
This completes the first step of the proof.

\medskip

\emph{Step 2.} We show that there exists an $\cF_\tau$-measurable function
\[
  \phi: \Omega\to L^1_{loc}(\bar{\Omega};\mathbb{S}^{>0}) \quad \mbox{such that} \quad P^{\phi(\omega)} = \nu(\omega) \quad \mbox{for $P$-a.e.\ }\omega\in\Omega.
\]
To this end, consider the set
\begin{equation*}
  A= \big\{(\omega,\alpha) \in \Omega \times L^1_{loc}(\bar{\Omega};\mathbb{S}^{>0}):\; \nu(\omega) = P^\alpha\big\} .
\end{equation*}
We have seen in Lemma~\ref{le:L1locBorelSpace} that $L^1_{loc}(\bar{\Omega};\mathbb{S}^{>0})$ is a Borel space.
On the other hand, we have from Lemma~\ref{le:PalphaBorel} that $\alpha\mapsto P^\alpha$ is Borel, and $\nu$ is Borel by assumption. Hence, $A$ is a Borel subset of
$\Omega \times L^1_{loc}(\bar{\Omega};\mathbb{S}^{>0})$. As a result, we can use the
Jankov--von Neumann theorem \cite[Proposition~7.49, p.\,182]{BertsekasShreve.78} to obtain an analytically measurable
map $\phi$ from the $\Omega$-projection of $A$ to $L^1_{loc}(\bar{\Omega};\mathbb{S}^{>0})$ whose graph is contained in $A$; that is,
\[
  \phi: \{\omega\in \Omega:\, \nu(\omega)\in \cP_S\} \to L^1_{loc}(\bar{\Omega};\mathbb{S}^{>0}) \quad \mbox{such that} \quad P^{\phi(\cdot)} = \nu(\cdot)\,.
\]
Since $\phi$ is, in particular, universally measurable, and since $\nu(\cdot)\in \cP_S$ $P$-a.s., we can alter $\phi$ on a $P$-nullset to obtain a Borel-measurable map
\[
  \phi: \Omega \to L^1_{loc}(\bar{\Omega};\mathbb{S}^{>0}) \quad \mbox{such that} \quad P^{\phi(\cdot)} = \nu(\cdot) \quad P\mbox{-a.s.}
\]
Finally, we can replace $\phi$ by $\omega\mapsto\phi(\omega_{\cdot\wedge \tau(\omega)})$, then we have the required $\cF_\tau$-measurability as a consequence of
Galmarino's test.
Moreover, since $A \in \cF_\tau\otimes \mathcal{B}(L^1_{loc}(\bar{\Omega};\mathbb{S}^{>0}))$ due to the $\cF_\tau$-measurability of $\nu$, Galmarino's test also shows that we still have
$P^{\phi(\cdot)} = \nu(\cdot)$ $P$-a.s., which completes the second step of the proof.

\medskip

\emph{Step 3.} It remains to construct $\hat{\alpha}\in L^0_{prog}(\bar{\Omega};\mathbb{S})$ as postulated in Step~1. While the map $\phi$ constructed in Step~2 will eventually yield the processes $\hat{\alpha}^\omega$ defined in Step~1, we note that $\phi$ is a map into a \emph{space of processes} and so we still have to glue its values into an actual process. This is simple when there are only countably many values; therefore, following a construction of~\cite{SonerTouzi.02b}, we use an approximation argument.

Since $L^1_{loc}(\bar{\Omega};\mathbb{S}^{>0})$ is separable (always for the metric introduced in~\eqref{eq:defConvLocallyInMeasure}), we can construct for any $n\in\N$ a countable Borel partition $(A^{n,k})_{k\geq 1}$ of $L^1_{loc}(\bar{\Omega};\mathbb{S}^{>0})$ such that the diameter of $A^{n,k}$ is smaller than $1/n$. Moreover, we fix $\gamma^{n,k}\in A^{n,k}$ for $k\geq1$. Then,
\[
  \phi_n(\omega) := \sum_{k\geq1} \gamma^{n,k} \1_{A^{n,k}}(\phi(\omega))
\]
satisfies
\begin{equation}\label{eq:phiUnivConv}
  \sup_{\omega\in \Omega } d(\phi_n(\omega),\phi(\omega))\leq \frac1n;
\end{equation}
that is, $\phi_n$ converges uniformly to $\phi$, as an $L^1_{loc}(\bar{\Omega};\mathbb{S}^{>0})$-valued map.

Let $\alpha$ and $\tilde{\tau}= \tau \circ X^\alpha$ be as in Step~1. Moreover, for any stopping time~$\sigma$, denote
\[
  \omega^\sigma_\cdot := \omega_{\cdot + \sigma(\omega)}-\omega_{\sigma(\omega)}, \quad \omega\in\Omega.
\]
Then, for fixed $n$, the process
\begin{align*}
  (\omega,s) \mapsto \hat{\alpha}^n_s(\omega)&:=\1_{[\tilde{\tau}(\omega),\infty)}(s)[\phi_n(X^\alpha(\omega))]_{s-\tilde{\tau}(\omega)}(\omega^{\tilde{\tau}(\omega)}) \\ &\phantom{:}\equiv \1_{[\tilde{\tau}(\omega),\infty)}(s)\sum_{k\geq1} \gamma_{s-\tilde{\tau}(\omega)}^{n,k}(\omega^{\tilde{\tau}(\omega)}) \1_{A^{n,k}}(\phi(X^\alpha(\omega)))
\end{align*}
is well defined $P_0$-a.s., and in fact an element of the Polish space  $L_{prog}^0(\bar{\Omega}; \mathbb{S})$. We show that $(\hat{\alpha}^n)$ is a Cauchy sequence and that the limit $\hat{\alpha}$ yields the desired process. Fix $T\in\R_+$ and $m,n\in\N$, then~\eqref{eq:phiUnivConv} implies that
\begin{equation}\label{eq:cauchy1}
  \int_\Omega\int_0^T 1\wedge \left|[\phi_m(\omega)]_s(\omega') - [\phi_n(\omega)]_s(\omega')\right| \, ds\, P_0(d\omega')\leq c_T \left(\frac1m + \frac1n\right)
\end{equation}
for all $\omega\in\Omega$, where $c_T$ is an unimportant constant coming from the definition of $d$ in~\eqref{eq:defConvLocallyInMeasure}.
In particular,
\begin{multline}\label{eq:cauchy2}
  \int_\Omega\int_\Omega\int_0^T 1\wedge \left|[\phi_m(X^\alpha(\omega))]_s(\omega') - [\phi_n(X^\alpha(\omega))]_s(\omega')\right| \, ds\, P_0(d\omega')\,P_0(d\omega)\\
  \leq c_T \left(\frac1m + \frac1n\right).
\end{multline}
Since $P_0$ is the Wiener measure, we have the formula
\[
  \int_\Omega g(\omega_{\cdot \wedge \tilde{\tau}(\omega)}, \omega^{\tilde{\tau}})\,P_0(d\omega) =   \int_\Omega\int_\Omega g(\omega_{\cdot \wedge \tilde{\tau}(\omega)}, \omega') \,P_0(d\omega')\,P_0(d\omega)
\]
for any bounded, progressively measurable functional $g$ on $\Omega\times\Omega$. As $\phi$ is $\cF_\tau$-measurable, we conclude from~\eqref{eq:cauchy2} that
\begin{equation}\label{eq:cauchy3}
  \int_\Omega\int_0^T 1\wedge \left|\hat{\alpha}^m_s(\omega) - \hat{\alpha}^n_s(\omega)\right| \, ds\,P_0(d\omega)
  \leq c_T \left(\frac1m + \frac1n\right).
\end{equation}
This implies that $(\hat{\alpha}^n)$ is Cauchy for the metric $d$.
Let $\hat{\alpha}\in L_{prog}^0(\bar{\Omega}; \mathbb{S})$ be the limit. Then, using again the same formula, we obtain that
\[
  \phi_n(X^\alpha(\omega)) = \hat{\alpha}^n_{\cdot+\tilde{\tau}(\omega)}(\omega\otimes_{\tilde{\tau}} \cdot ) \to \hat{\alpha}_{\cdot+\tilde{\tau}(\omega)}(\omega\otimes_{\tilde{\tau}} \cdot ) \equiv \hat{\alpha}^\omega
\]
with respect to $d$, for $P_0$-a.e.\ $\omega\in\Omega$, after passing to a subsequence. On the other hand,
by~\eqref{eq:phiUnivConv}, we also have $\phi_n(X^\alpha(\omega))\to \phi(X^\alpha(\omega))$ for $P_0$-a.e.\ $\omega\in\Omega$. Hence,
\[
  \hat{\alpha}^\omega = \phi(X^\alpha(\omega))
\]
for $P_0$-a.e.\ $\omega\in \Omega$. In view of Step~2, we have $\phi(X^\alpha(\omega))\in L^1_{loc}(\bar{\Omega};\mathbb{S}^{>0})$ and $P^{\phi(X^\alpha(\omega))}=\nu(X^\alpha(\omega))$ for $P_0$-a.e.\ $\omega\in \Omega$. Hence, $\hat{\alpha}$ satisfies all requirements from Step~1 and the proof is complete.

\section{Proof of Theorem~\ref{th:duality}}\label{se:proofDuality}

We note that one inequality in Theorem~\ref{th:duality} is trivial: if $x\in\R$ and there exists $H\in\cH$ such that $x+ \int_0^T H\,dB \geq \xi$, the supermartingale property stated in the definition of $\cH$ implies that $x\geq E^P[\xi]$ for all $P\in\cP$. Hence, our aim in this section is to show that there exists $H\in\cH$ such that
\begin{equation}\label{eq:proofDualityAim}
  \sup_{P\in\cP} E^P[\xi] + \int_0^T H_u\,dB_u \geq \xi \quad P \mbox{-a.s.} \quad \mbox{for all} \quad P \in \mathcal{P}.
\end{equation}
The line of argument (see also the Introduction) is similar as in \cite{SonerTouziZhang.2010dual} or \cite{NutzSoner.10}; hence, we shall be brief.

We first recall the following known result (e.g., \cite[Theorem~1.5]{JacodYor.77}, \cite[Lemma~8.2]{SonerTouziZhang.2010aggreg}, \cite[Lemma~4.4]{NutzSoner.10}) about the $P$-augmentation $\overline{\F}^P$ of $\F$; it is the main motivation to work with $\cP_S$ as the basic set of scenarios. We denote by $\G^+=\{\cG_{t+}\}_{0\leq t\leq T}$ the minimal right-continuous filtration containing $\G$.

\begin{lemma}\label{le:MRPandVersions}
  Let $P\in\cP_S$. Then $\overline{\F}^P$ is right-continuous and
  in particular contains $\G^+$. Moreover,
  $P$ has the predictable representation property; i.e., for any right-continuous $(\overline{\F}^P,P)$-local martingale $M$ there exists an $\overline{\F}^P$-predictable process $H$ such that
  $M=M_0+\int H\,dB$, $P$-a.s.
\end{lemma}

We recall our assumption that $\sup_{P\in\cP} E^P[|\xi|]<\infty$ and that $\xi$ is $\cG_T$-measurable. We also recall from Proposition~\ref{pr:NvH} that the random variable
\[
  \cE_t(\xi)(\omega):=\sup_{P\in\cP(t,\omega)} E^P[\xi^{t,\omega}]
\]
is $\cG_t$-measurable for all $t\in\R_+$. Moreover, we note that $\cE_T(\xi)=\xi$ $P$-a.s.\ for all $P\in\cP$. Indeed, for any fixed $P\in\cP$, Lemma~\ref{le:MRPandVersions} implies that we can find an $\cF_T$-measurable function $\xi'$ which is equal to $\xi$ outside a $P$-nullset $N\in\cF_T$, and now the definition of  $\cE_T(\xi)$ and Galmarino's test show that $\cE_T(\xi)=\cE_T(\xi')=\xi'=\xi$ outside $N$.

\medskip

\emph{Step 1.} We fix $t$ and show that $\sup_{P\in\cP} E^P[|\cE_t(\xi)|]<\infty$. Note that $|\xi|$ need not be upper semianalytic, so that the claim does not follows directly from~\eqref{eq:DPP}. Hence, we make a small detour and first observe that $\cP$ is stable in the following sense: if $P\in\cP$, $\Lambda\in\cF_t$ and $P_1,P_2\in\cP(t;P)$  (notation from Proposition~\ref{pr:NvH}), the measure $\bar{P}$ defined by
\[%
  \bar{P}(A):=E^P \big[P_1(A|\cF_t)\1_\Lambda + P_2(A|\cF_t)\1_{\Lambda^c}\big],\quad A\in\cF
\]%
is again an element of $\cP$. %
Indeed, this follows from (A2) and (A3) as
\begin{align*}
\bar{P}(A)
&= \iint (\1_A)^{t, \omega}(\omega')\, \nu(\omega, d\omega')\, P(d\omega)
\end{align*}
for the kernel $\nu(\omega,d\omega')=P_1^{t, \omega}(d\omega')\,\1_{\Lambda}(\omega) +
P_2^{t, \omega}(d\omega')\,\1_{\Lambda^c}(\omega)$. Following a standard argument, this stability implies that for any $P\in\cP$,
there exist $P_n\in\cP(t;P)$ such that
\[
  E^{P_n}[|\xi|\,|\cF_t] \nearrow {\mathop{\esssup^P}_{P'\in \cP(t;P)}} E^{P'}[|\xi|\,|\cF_t] \quad P\mbox{-a.s.}
\]
Since~\eqref{eq:esssupDPP}, applied with $\tau=T$, yields that
\begin{align*}
  E^P[|\cE_t(\xi)|]
   = E^P\bigg[\bigg|{\mathop{\esssup^P}_{P'\in \cP(t;P)}} E^{P'}[\xi|\cF_t]\bigg|\bigg]
   \leq E^P\bigg[{\mathop{\esssup^P}_{P'\in \cP(t;P)}} E^{P'}[|\xi|\,|\cF_t]\bigg],
\end{align*}
monotone convergence then allows us to conclude that
\begin{align*}
  E^P[|\cE_t(\xi)|]
  \leq \lim_{n\to\infty} E^{P_n}[|\xi|]
  \leq \sup_{P\in\cP}E^{P}[|\xi|]<\infty.
\end{align*}

\emph{Step 2.} We show that the right limit $Y_t:=\cE_{t+}(\xi)$ defines a $(\G^+,P)$-supermartingale for all $P\in\cP$. Indeed, Step~1 and~\eqref{eq:esssupDPP} show that $\cE_t(\xi)$ is an $(\F^*,P)$-supermartingale for all $P\in\cP$. The standard modification theorem for supermartingales \cite[Theorem~VI.2]{DellacherieMeyer.82} then yields that $Y$ is well defined $P$-a.s.\ and that $Y$ is a $(\G^+,P)$-supermartingale for all $P\in\cP$, where the second conclusion uses Lemma~\ref{le:MRPandVersions}. We omit the details; they are similar as in the proof of~\cite[Proposition~4.5]{NutzSoner.10}.

For later use, let us also establish the inequality
\begin{equation}\label{eq:ineqTimeZero}
  Y_0 \leq \sup_{P'\in\cP} E^{P'}[\xi] \quad P \mbox{-a.s.} \quad \mbox{for all} \quad P \in \mathcal{P}.
\end{equation}
Indeed, let $P\in\cP$. Then \cite[Theorem~VI.2]{DellacherieMeyer.82} shows that
\[
 E^P[Y_0|\cF_0]\leq \cE_0(\xi) \quad P\as,
\]
where, of course, we have  $E^P[Y_0|\cF_0]=E^P[Y_0]$ $P$-a.s.\ since $\cF_0=\{\emptyset,\Omega\}$. However, as $Y_0$ is $\cG_{0+}$-measurable and $\cG_{0+}$ is $P$-a.s.\ trivial by Lemma~\ref{le:MRPandVersions}, we also have that $Y_0=E^P[Y_0]$ $P$-a.s. In view of the definition of $\cE_0(\xi)$, the inequality~\eqref{eq:ineqTimeZero} follows.

\medskip

\emph{Step~3.} Next, we construct the process $H\in\cH$. In view of Step~2, we can fix $P\in\cP$ and consider the Doob--Meyer decomposition $Y=Y_0+M^P-K^P$ under $P$, in the filtration $\overline{\F}^P$. By Lemma~\ref{le:MRPandVersions}, the local martingale $M^P$ can be represented as an integral,
$M^P=\int H^P\,dB$, for some $\overline{\F}^P$-predictable integrand $H^P$. The crucial observation (due to \cite{SonerTouziZhang.2010dual}) is that this process can be described via $d\br{Y,B}=H^P\,d\br{B}$, and that, as the quadratic covariation processes can be constructed pathwise by Bichteler's integral~\cite[Theorem~7.14]{Bichteler.81}, this relation allows to define a
process $H$ such that $H=H^P$ $P\times dt$-a.e.\ for all $P\in\cP$. More precisely, since $\br{Y,B}$ is continuous, it is not only adapted to $\G^+$, but also to $\G$, and hence we see by going through the arguments in the proof of \cite[Proposition~4.11]{NutzSoner.10} that $H$ can be obtained as a $\G$-predictable process in our setting. To conclude that $H\in\cH$, note that for every $P\in\cP$, the local martingale $\int H\,dB$ is $P$-a.s.\ bounded from below by the martingale $E^P[\xi|\G]$; hence, on the compact $[0,T]$, it is a supermartingale as a consequence of Fatou's lemma.
Summing up, we have found $H\in\cH$ such that
\[
  Y_0 + \int_0^T H_u\,dB_u \geq Y_T=\cE_{T+}(\xi)=\xi \quad P\mbox{-a.s.} \quad \mbox{for all} \ P \in \cP,
\]
and in view of~\eqref{eq:ineqTimeZero}, this implies~\eqref{eq:proofDualityAim}.

\newcommand{\dummy}[1]{}

\end{document}